\documentclass[aps,pre,reprint,superscriptaddress]{revtex4-2}

\usepackage[T1]{fontenc}
\usepackage{mathptmx}
\usepackage{amsmath,amssymb,amsthm,bm}
\usepackage{graphicx}
\usepackage{booktabs}
\usepackage{multirow}
\usepackage{natbib}

\usepackage[version=3]{mhchem}
\usepackage{chemformula}
\usepackage{orcidlink}

\usepackage{hyperref}
\usepackage[normalem]{ulem}

\newtheorem{theorem}{Theorem}

\newtheorem{remark}{Remark}

\newcommand*{\addref}[1]{}
\newcommand*{\noter}[1]{}
\newcommand*{\pagelimit}[1]{}

\newcommand*{\weicomm}[1]{}
\newcommand*{\sidcomm}[1]{}
\newcommand{\added}[1]{#1}
\newcommand{\deleted}[1]{}

\newcommand{\reviewnote}[1]{}

\begin{document}

\preprint{APS/123-QED}

\title{Kolmogorov--Sinai entropies identify optimal observables for prediction and dynamics reconstruction in chaotic systems}

\author{Maximilian Topel\,\orcidlink{0000-0003-4374-7668}}
\email{maxtopel@northwestern.edu}
\affiliation{
 Department of Applied Mathematics, Northwestern University,  145 Sheridan Rd, Evanston, IL 60208, USA
}

\begin{abstract}
Choosing the optimal observable to model dynamical systems for which we do not know the driving equations is nearly always an ad hoc art. Embedding theorems provide theoretical insight into the relationship between time series observations of a system and the dynamics that generate them. Takens' Delay Embedding Theorem guarantees a diffeomorphism between delay-coordinate vectors built from generic scalar observables and the underlying invariant attractors from which they were generated. However, Takens' Theorem is agnostic to the choice of observable, and formal bounds on the quality of reconstruction across observables are not known. Here we prove that, under modest technical conditions, the Kolmogorov-Sinai entropy of an observable predicts its reconstruction error of the underlying dynamics in chaotic, ergodic systems. Using the Oseledets Multiplicative Ergodic Theorem, we demonstrate that the tangent bundles of the reconstructed manifolds admit an invariant Oseledets filtration diffeomorphically related across admissible observables. The Lyapunov exponents associated with each subspace then control the propagation of perturbations through the manifold as a whole. We bound reconstruction errors as a function of the contributions from each Oseledets subspace, showing that they are bounded above by a quantity monotonically related to the sum of positive Lyapunov exponents and, by the Ruelle inequality, the Kolmogorov-Sinai entropy. We validate this finding  empirically on the Lorenz-63 attractor, the Hastings--Powell food chain, and a tetracosane molecular-dynamics trajectory, recovering Spearman rank correlations between $h^{KS,UB}$ and reconstruction RMSE up to $\rho=+0.89$ ($p=5.5\times 10^{-8}$) for the realistic tetracosane system, sharpening  $\rho=+0.97$ under added measurement noise reflecting the asymmetric nature of noise on observables of the same dynamics. This work provides a rigorous foundation for observable selection in chaotic systems, applicable as an \emph{a priori} data-selection criterion for any data-driven modeling pipeline.

\end{abstract}


\maketitle

\section{Introduction}

Advances in the fields of complex systems \cite{strogatz,Barabasi2002,Mitchell2009} and machine learning \cite{Goodfellow2015} alike have enabled the study of complicated, interconnected and noisy systems from population dynamics \cite{May1976} to consumer demand functions \cite{Varian2014}. These tools have put the study and prediction of complex systems comfortably within the reach of anyone with access to compute. However, these marked advances in learning technologies have not, by and large, made use of the rich theoretical and analytical frameworks built by dynamical systems theorists from the 1970s-1990s in spite of their natural connection with modern learning problems. Data-driven learning of dynamics and control have become increasingly popular both mathematically and computationally through methods like recurrent neural networks (Hamilton-Jacobi Bellman equations), Koopman/Dynamic Mode Decomposition frameworks and SINDy \cite{Mezic2005spectral,Williams2015EDMD,Brunton2016SINDY}. However, such efforts have not been met with commensurate efforts to optimize the basis choice upon which we build such models. Within a sea of plausible time series observations of dynamics coming from any number of data sources, deciding which to use can be tricky. Some works have considered explicit symbolic-Jacobian observability coefficients \cite{letellier2002investigating,letellier2006how,aguirre2005observability,aguirre2009modeling} but these approaches do not extend to situations where the underlying driving equations are not known. 

Which observables do we choose to do modeling? If dynamical systems encompass every process evolving in time according to a rule set, our field reaches far beyond systems for which we have known ODEs. Modeling often uses heuristic or intuitive choices but in the face of ever more complex dynamics, we need principled ways of selecting the observations we use to make models. In this work, we supply a principled answer rooted in classical ergodic theory and the geometrical structure inherent in time series. While embedding theorems provide us with the recipes for constructing geometry from time series observables, they do not tell us what choice is optimal. Moreover, the individual nature of each and every observable choice opens the door to bending, squashing, shearing and twisting of such structures meaning that distances on such geometric objects are not guaranteed to be the same. We propose how and prove why considering a measure of these deformations, the Kolmogorov-Sinai Entropy (KSE), can be used to predictively select optimal time series to feed into machine learning platforms modeling chaotic ergodic systems.

In previous work, it has been shown that Kolmogorov-Sinai entropies have strong physical relations to physical and statistical quantities describing a dynamical system. In Refs. \cite{PhysRevLett.82.520,PhysRevE.62.6516} it is shown that the Kolmogorov-Sinai entropy is related to thermodynamic entropy in molecular systems. Ref. \cite{PhysRevLett.81.1762} discusses the mathematical relationship between statistical entropy and the Kolmogorov-Sinai entropy in chaotic dynamical systems more generally. Ref. \cite{PhysRevE.71.046211} shows the Kolmogorov-Sinai entropies correspond to the average of the logarithm of the rate of expansion part of tangent space (stretching factor). These innovations have demonstrated that the Kolmogorov Sinai entropies are mathematically tied to geometrical and statistical descriptions of dynamics. In this work we build upon this understanding and formalism to bound the quality of invariant manifold or attractor reconstruction and consequently backmapping.

We apply a combination of geometry, embedding theorems and dynamics to first define the problem of time series selection, introduce the notion of geometries from time series, discuss differences in the manifolds describing dynamics of each and finally prove how the KSE can be used to optimize time series selection. In Sect. \ref{KSE:sec:ergodic}, we introduce the notion of ergodic dynamics and invariant manifolds arising from complex dynamics. In Sect. \ref{KSE:sec:takens}, we discuss Takens' Delay Embedding Theorem and its application in ergodic dynamics. In Sect. \ref{KSE:sec:distort}, we describe types of geometrical distortion to attractors constructed from time series. In Sect. \ref{KSE:sec:omet}, we introduce Oseledets Multiplicative Ergodic Theorem and its application to Takens embeddings under certain conditions. In Sect. \ref{KSE:sec:lyap}, we define Lyapunov exponents and the Kolmogorov-Sinai entropy. In Sect. \ref{KSE:sec:TandP}, we prove that under certain modest technical conditions, Kolmogorov-Sinai entropy predictively selects time series that minimize reconstruction errors when backmapping to full dimensional dynamics. Finally, in Sect. \ref{KSE:sec:empirical}, we demonstrate this theoretical effect on three dynamical systems of varying complexity to illustrate its efficacy in selecting observables for reconstruction.

\section{Ergodic Dynamics and Attractors} \label{KSE:sec:ergodic}
%define ergodicity
We can define the accessible phase space of a dynamical system $\mathbb{S}$ in terms of its measurable phase space given by the triple $(X,\Sigma,\mu)$ where $X$ defines accessible states, $\Sigma$ specifies the $\sigma$-algebra consisting of all measurable subsets in $X$, and $\mu$ is the probability measure defining the probability of measurable sets in $X$ \cite{Walters1982,CornfeldSinaiFomin1982}. We can also define some measure preserving transformation $T: X\rightarrow X$ that propagates $X$ through $\mathbb{S}$. $T$ is said to be ergodic if for every invariant set $A\in\Sigma$ satisfying $T^{-1}(A)=A$, $\mu(A) = 0$ or $\mu(A) = 1$ \cite{Walters1982}.

%Liouville
For instance, in classical mechanics, Liouville's Theorem guarantees that there exists a time evolution map given by Hamilton's equations of motion where phase space volume is preserved under the time evolution \cite{Arnold1989}. More formally, for any measurable set $A\subseteq X$, $\mu(A)$ remains constant under time evolution, $\mu(T^t(A))$. This transformation is said to be symplectic (volume-preserving) \cite{Arnold1989}. It may also be ergodic with respect to the Liouville measure; however, constraints on dynamics can result in the emergence of restrictions to accessible phase space.

%introduce attractors
The existence of conserved quantities such as energy in classical physical systems results in the emergence of lower dimensional surfaces (invariant manifolds) in phase space that can effectively describe system dynamics \cite{Noether1918}. For instance, an N-atom molecular system evolving in some high dimensional phase space $\mathbb{R}^{6N}$ has an effective description on an invariant manifold of $k\ll6N$ dimensions. These dynamics are restricted to $\mathbb{R}^{6N-1}$ by energy constraints and then can be further confined to $\mathbb{R}^{k}$ as a result of atomic interactions such as electrostatic forces \cite{Wales2003}.  We term this subset of phase space the \textit{attractor}, although this term is also used for manifolds emerging from dissipative dynamics where trajectories converge. We extend this usage to describe any low dimensional and invariant subset of an original high dimensional description. In this work, we consider the broader context of both dissipative and chaotic systems giving rise to attractors. Such systems may have an invariant Sinai--Ruelle--Bowen (SRB) measure describing the statistical behavior of chaotic attractors which are typically ergodic on systems where they are defined \cite{Young2002,Ruelle1979}. These attractors have been shown to be robust to perturbations as in Kolmogorov--Arnold--Moser (KAM) theory \cite{Arnold1989} and in the works of Stark and Broomhead under both stochastic and deterministic forcing \cite{stark1999delay,stark2003delay,Broomhead1986}. If we constrain $X$ to these invariant attractors, we can often define some transformation $T$ that is ergodic with respect to a measure $\mu$ on these invariant manifolds \cite{Walters1982,RuelleEckmann}.

\section{Takens' Delay Embedding Theorem}  \label{KSE:sec:takens}

%introduce Takens
Takens' Delay Embedding Theorem states that for any dynamical system, $\mathbb{S}$, with smooth dynamics describable by a compact attractor (i.e. on an invariant manifold), for which we can record a time series observation of its dynamics, vectors constructed by recording time delayed scalar observables of $\mathbb{S}$ constitute an embedding of the attractor of $\mathbb{S}$ onto $\mathbb{R}^{2m+1}$ where $m$ is the dimensionality of the attractor under modest technical conditions \cite{Takens}. These vectors have form $y(t) = [\mathrm{o}_t,\mathrm{o}_{t-\tau}, \dots, \mathrm{o}_{t-2m\tau}] \in \mathbb{R}^{2m+1}$ and their embedding describes the geometry and dynamical properties of $\mathbb{S}$ up to an \textit{a priori} unknown diffeomorphism \cite{Takens}.

This powerful theorem can be used to extract robust descriptions of a dynamical system from time series observations of that system. This is particularly convenient for complex dynamics where the forces or rules that propagate that system are unknown, or for large interacting systems like molecules, where propagating those rules in a full-dimensional form is very expensive. While this theorem allows us to identify these intrinsic manifolds describing dynamics, these manifolds do not necessarily have ergodic measures. There may exist multiple invariant components to these manifolds and ergodicity requires that these systems not be decomposable. That is to say, for an ergodic measure $\mu$ on the map $T$ describing all accessible states $X$ for $(X,\Sigma,\mu)$, there cannot be some  $A \subset X$ that is not accessible by the propagation of any other subset $A^{\prime}\subset X$, $T^t(A^{\prime})$. If there were multiple closed sets in $X$, $X$ could be subdivided into multiple subsets where $\mu(A)= 0,1$, violating ergodicity. In this case, an invariant measure could be composed on each closed subset of $X$, however the decomposability would violate the requirements of ergodicity.

Here we will appeal to two facts. 1) If we collect continuously observed time series observations of a dynamical system $S$ then we typically only capture the states within an invariant subset of the phase space \cite{strogatz}. We can then apply ergodicity to this subset of phase space if there is an ergodic invariant measure on it. It should be noted that stochastic or deterministic forcing causing a jump to a different invariant manifold is possible in the case where there exist multiple invariant manifolds describing dynamics for forced systems \cite{Grebogi1983}.  2) It is often the case that the chaotic dynamical systems we discuss in this work enjoy an invariant Sinai--Ruelle--Bowen measure describing that is typically ergodic on systems where they are defined \cite{Young2002}. This SRB measure is considered to be physical because it describes statistical behavior congruent on almost any initial conditions characteristic of the attractor \cite{Young2002,RuelleEckmann}.

\section{Distortions to Embedded Manifolds}  \label{KSE:sec:distort}

 Our discussion of both Takens' Theorem and dynamics has been somewhat idealized. In the real world, our time series observations may be subject to noise or deterministic forcing. Moreover, chaotic systems are subject to "deterministic noise", small differences in initial conditions resulting in exponentially different trajectories over time \cite{RuelleEckmann}. Consequently, although a system may evolve deterministically, the time series observations of that system may resemble stochastically defined noise. This high degree of sensitivity to initial conditions for chaotic systems means that any imprecision in our definition of a state $x \in X$ can lead to a mischaracterization of the embedding $\mathcal{M}\subseteq\mathbb{R}^{2m+1}$ we can produce from Takens' Theorem. Possible sources of error abound. Stochastic noise from recording equipment, deterministic forcing from a second and separable form of dynamics influencing our system or even errors determining the correct value of $x$ due to errors accrued solving differential equations numerically over long time scales can limit our ability to correctly and uniquely define the set of accessible states $X$ on an attractor over time \cite{KantzSchreiber2004}.

It is fortunate, then, that the works of Stark and Broomhead  \cite{Broomhead1986,stark1999delay,stark2003delay} have demonstrated the robustness of Takens' Theorem to such deterministic and stochastic forcing  as well as impact of noise on state space reconstruction in Casdagli \emph{et al.}\ \cite{Casdagli1991statespace} and Schreiber and Grassberger \cite{SchreiberGrassberger1991}. However, we expect that the degree of sensitivity of any one of our potential time series observations to noise will result in distortions to the embedding $\mathcal{M}^{\prime}\subseteq\mathbb{R}^{2m+1}$ that will be a function of individual manifold topography. Consequently, while Takens' Theorem applies to generic observables, in practice, observable choice can affect the quality of modeling performed using Takens delay vectors constructed from that time series \cite{KantzSchreiber2004,Takens,sauer1991embedology}. While observable selection methodologies \cite{letellier2002investigating,letellier2006how,aguirre2005observability,aguirre2009modeling}, formal bounds on reconstruction error rooted in the dynamics exhibited by these observables are, to our knowledge, absent.

We can, however, classify the types of distortions to latent space representations of dynamics that can arise in manifolds constructed from Takens' Theorem. While Takens delay embeddings are guaranteed to be diffeomorphically related to manifolds describing full dimensional dynamics, they are subject to smooth deformations of their topography including stretching, squashing, bending, local warping, twisting and shearing \cite{HirschSmaleDevaney2003}. These will result in deformations to the intrinsic geometry of the manifold, meaning that there will be changes to the distances between embedded points on $\mathcal{M}'$ produced by Takens' Embedding.

\begin{table*}[ht!]
\centering
\resizebox{\textwidth}{!}{
\begin{tabular}{l p{8cm} c}
\hline
\textbf{Deformation Type} & \textbf{Description} & \textbf{Intrinsic Geometry Conserved} \\
\hline
Stretching/Squashing & Expansion or compression of regions of phase space. & No \\
Bending & Smooth curving that does not induce folds. & Yes \\
Local Warping & Small deformations on the surface. & No \\
Twisting & Smooth rotation. & Yes \\
Shearing & Deformation in angles and distances. & No \\
\hline
\end{tabular}
}
\caption{Types of deformations of the manifold $\mathcal{M}'$ under $C^1$ equivalence as guaranteed by Takens' Theorem that may vary with observable choice \cite{HirschSmaleDevaney2003,Takens,sauer1991embedology}.}
\label{KSE:table:deformations}
\end{table*}

Any \textit{intrinsic} geometrical deformation alters distances on the manifold while \textit{extrinsic} ones may visually affect the manifold but leave distances unchanged \cite{HirschSmaleDevaney2003}. Minimizing intrinsic distortions leads to more faithful representations of system dynamics. Moreover, larger deformations present the opportunity for greater sensitivity to mischaracterization of $r^{\prime}(t)\in\mathcal{M}^{\prime}$ and consequently, observable selection should be done in a fashion that minimizes such deformations. In the following section, we apply Oseledets Multiplicative Ergodic Theorem to study the rate of such distortions.

\section{Oseledets Multiplicative Ergodic Theorem} \label{KSE:sec:omet}

It is now incumbent upon us to study the rate at which the reconstructed manifold $\mathcal{M}^{\prime}\subseteq \mathbb{R}^k$ produced by Takens' Theorem is distorted. We assume first that our observations describe an effective phase space described by an attractor, $\mathcal{M} \subseteq \mathbb{R}^k$, and that the time evolution of state $r(t)\in \mathcal{M}$ is described by transformation $T$ ergodic on the invariant measure $\mu$. We contend that this is a fair assumption for a wide variety of chaotic dynamical systems as discussed in the previous sections. Moreover, ergodicity ensures that time averages converge to ensemble averages via Birkhoff's Ergodicity Theorem, ensuring reliability of the statistical properties derived from observations \cite{Birkhoff1931}.

Oseledets multiplicative ergodic theorem (OMET) provides the convergence criterion in the rates of growth or shrinkage of the tangent space (or bundle) at almost every point on the attractor $\mathcal{M} \subseteq \mathbb{R}^{2m+1}$ \cite{Oseledets1968}. It is multiplicative as it deals with the product of linear operators and is vector-valued; that is to say it describes time averaged behavior of the evolution of a vector under a sequence of linear transformations to our system \cite{Arnold1998}. For some linear operator $A(r)$ operating on the orbit $T^j(r)$ acting on vector $r(t)\in\mathcal{M}\subseteq\mathbb{R}^{2m+1}$, there exist statistical limits to the evolution of $r$ propagated by $A_t(r) = A\bigl(T^{t-1}(r)\bigr) \cdots A\bigl(T(r)\bigr) A(r)$. This operator corresponds to the cocycle over the base dynamics, $T$.  OMET tells us that we can subdivide the tangent space of the attractor $\mathcal{M} \subseteq \mathbb{R}^{2m+1}$ into a filtration of subspaces $\{0\}\subseteq V_0\subseteq V_1\subseteq...\subseteq V_q\subseteq V_{q+1}\subseteq...\subseteq V_K\subseteq\mathbb{R}^{2m+1}$  for which there exist a corresponding set of Lyapunov exponents $ \lambda_1\geq \lambda_2 \geq ... \geq \lambda_K$ measuring rate of growth or shrinkage of each subspace under  the transformation $A$ \cite{Oseledets1968,Arnold1998}.  Moreover, these rates of growth and shrinkage are given by the limit behavior of the evolution of our system at long timescale:

\begin{equation}
    \lim_{t\rightarrow\infty}\dfrac{1}{t}\log\|A_t(r)v\|= \lambda_q
\end{equation}

where $v \in V_q\backslash V_{q-1}$ is a nonzero vector in subspace $V_q$ but not $V_{q-1}$ \cite{Oseledets1968,Arnold1998}.

We can equivalently apply OMET to some manifold extracted via Takens' Delay Embedding Theorem $\mathcal{M}^{\prime} \subseteq \mathbb{R}^{2m+1}$ that enjoys a smooth, invertible and bijective map between manifolds, $\Phi:\mathcal{M}^{\prime}\mapsto\mathcal{M}$, leading to its own filtration of subspaces $\{0\}\subseteq V_0^{\prime}\subseteq V_1^{\prime}\subseteq...\subseteq V^{\prime}_q\subseteq V^{\prime}_{q+1}\subseteq...\subseteq V_K^{\prime}\subseteq\mathbb{R}^{2m+1}$  for which there exist a corresponding set of Lyapunov exponents $ \lambda^{\prime}_1\geq \lambda^{\prime}_2 \geq ... \geq \lambda^{\prime}_K$ produced by the limiting behavior:

\begin{equation}
    \lim_{t\rightarrow\infty}\dfrac{1}{t}\log\|A^{\prime}_t(r^{\prime})v^{\prime}\|= \lambda^{\prime}_q
\end{equation}

where $r^{\prime}$ is a point on the reconstructed manifold $\mathcal{M}^{\prime}$, $A^{\prime}_t(r^{\prime})$ describes the linear map of propagating $r^{\prime}$ over base dynamics $T^{\prime}$ with orbit $T^{j\prime}$ and $v^{\prime} \in V_q^{\prime}\backslash V_{q-1}^{\prime}$ is a vector in the tangent space at $r^{\prime}$.

Takens theorem guarantees that the manifolds $\mathcal{M}$ and $\mathcal{M}^{\prime}$ are diffeomorphically related \cite{Takens}. As we have assumed that base dynamics $T$ and $T^{\prime}$ are ergodic on the invariant measure $\mu$ and $\mu^{\prime}$ respectively, we can apply OMET. Moreover, there exists a bijective map between any given point on the orbits $T^{j}(r)$ and $T^{j\prime}(r^{\prime})$ in the limit that the requirements of Takens' theorem are fulfilled. The induced diffeomorphism between $\mathcal{M}$ and $\mathcal{M}^{\prime}$ holds for their tangent bundles; this ensures invariant splittings defined by the filtration in OMET, $V_q$ and $V_q^{\prime}$, also enjoy a diffeomorphic relation.

Earlier, we argued that the fidelity of latent space manifold extracted via Takens' theorem, $\mathcal{M}^{\prime}$, is a function of stochastic and deterministic forcing as well as inconsistencies in the recording of time series $\mathrm{O}$. Applying this to our understanding of OMET, it follows that reconstruction of any given invariant subspace $V_q$ will then be a function of the deformation to the intrinsic geometry of each subspace induced by noise and forcing. It then follows that the sensitivity of each subspace $V_q^{\prime}$ to infinitesimal errors in our estimation of $v_q^{\prime} \in V_q^{\prime}$ is given by $\lambda_q^{\prime}$. This fact is the basis of our proof in Section \ref{KSE:sec:TandP}. In the next section we explore the interpretation of these Lyapunov exponents and the aggregate effect of their values across all subspaces.

\section{Lyapunov Exponents and the Kolmogorov Sinai Entropy}  \label{KSE:sec:lyap}

The Lyapunov exponents of a dynamical system describe the average exponential rate of divergence or convergence of infinitesimally close trajectories on a phase space describing the evolution of a dynamical system. Empirically, they are extractable by fitting a function of mean distance between trajectories $\delta(t)=\delta(0)e^{\lambda t}$ where $\delta(t)$ describes the mean distance between trajectories at time $t$ with Lyapunov exponent $\lambda$. Positive Lyapunov exponents correspond to exponential divergence in trajectories, negative ones to convergence or dissipation while a zero Lyapunov exponent typically characterizes motion along a limit cycle attractor \cite{LyapEckmann,rosenstein93}.

In the previous section, we discussed how the rate of growth and shrinkage of subspaces for systems satisfying OMET are related to these exponents. They can be extracted directly from time series using the QR-based method of Eckmann \emph{et al.}\ \cite{LyapEckmann} or the small-divergence method of Rosenstein \emph{et al.}\ \cite{rosenstein93}. Now we consider the aggregate effect of these deformations on the time series.

The \textbf{Kolmogorov--Sinai entropy (KSE), $h^{\mathrm{KS}}$,} quantifies the average rate of information production in the system. It is defined as
\begin{equation}
h^{\mathrm{KS}} = \lim_{\epsilon \to 0} \lim_{n \to \infty} \frac{1}{n} H(\epsilon, n),
\end{equation}
where \( H(\epsilon, n) \) is the Shannon entropy of trajectories defined with resolution \( \epsilon \) over \( n \)-step intervals \cite{Walters1982}. However, it is also extractable directly from the Lyapunov exponents. Intuitively, this is because positive Lyapunov exponents describe chaotic behavior, and consequently information production. Ruelle's inequality allows us to bound this directly from the positive Lyapunov exponents of our time series \cite{Ruellesinequality}:

\begin{equation}
   h^{\mathrm{KS}} \leq \sum_{\lambda_i \geq 0} \lambda_i.
\end{equation}

Under appropriate conditions (e.g., for smooth hyperbolic systems with an SRB measure), Pesin's Theorem shows that the Kolmogorov--Sinai entropy equals the sum of positive Lyapunov exponents  \cite{PesinTheorem,Young2002}~\cite{LedrappierYoung1985}.

As the upper bound limit of the KSE is computed from the Lyapunov exponents of a given time series and these are computed by the mean divergence in the distance of infinitesimally close points on the manifold extracted from this time series ($\mathcal{M}^{\prime}$), we hypothesize that manifolds with larger KSEs will be more sensitive to noise, forcing and errors in specifying a given $r^{\prime}\in\mathcal{M}^\prime$\added{.}

\section{Theorem and Proof \label{KSE:sec:TandP}}

Here we prove that time series exhibiting smaller Kolmogorov-Sinai entropies can produce reconstruction of full dimensional dynamics with smaller mean errors than those with larger Kolmogorov--Sinai entropies.

\paragraph{Setup and Assumptions}
Let $\mathbb{S}$ be a dynamical system whose dynamics are chaotic and evolve on the compact invariant  attractor $\mathcal{M}\subseteq\mathbb{R}^{2m+1}$ of box-counting dimension $d_\mathcal{M}$. We fix the embedding dimensionality $m\geq d_\mathcal{M}$ so that $\mathcal{M}$ admits a generic embedding into $\mathbb{R}^{2m+1}$ by Takens' Theorem. Let $\mathrm{O}_i(t), i\in \{1,...,n\} \in{\mathbb{N}}$ be a set of $n$ true scalar time series observables of $\mathbb{S}$. Due to imperfect measurement, external forcing and stochastic noise, the observed time series is then $\tilde{\mathrm{O}}_i(t)=\mathrm{O}_i(t)+\varepsilon_i(t)$ where $\varepsilon_i(t)$ represents the perturbation to observable $i$ at time $t$, with each $\varepsilon_i$ drawn from $\Omega$, noise that is independent of the dynamics on $\mathcal{M}$.

We assume:

\begin{enumerate}
    \item $\varepsilon(t)$ is drawn from some probability distribution $\Omega$. $\Omega$ can be time dependent in the case of deterministic forcing. Let $\|\varepsilon(t)\| \ll \| \mathrm{O}(t)\| $ and consequently the signal to noise ratio, $\mathrm{SNR}:=\sigma_{\mathrm{O}_i}^2/\sigma_{\varepsilon_i}^2\gg 1$. Let this noise be uncoupled with the nonlinear dynamics of $\mathcal{M}.$ $\varepsilon_i$ is statistically independent of the state $r(t)\in\mathcal{M}$ to leading order (weak state-dependence at $\mathcal{O}(\sigma_{\varepsilon_i}^2)$). 
    \item The trajectories of each $\mathrm{O}_i(t)$ produce embeddings $\mathcal{M}_{i}\subseteq\mathbb{R}^{2m+1}$ via Takens' Theorem upon which the time-evolved trajectories on the embedded manifolds, $T_{i}(t)$, exist. \item All of the requirements of Takens' Theorem be fulfilled including sufficient sampling time such that all $\mathcal{M}_i$ will be diffeomorphically related to each other and the intrinsic attractor $\mathcal{M}$.
    \item Each orbit $T_{i}(t)$ on $\mathcal{M}_i$ be ergodic with respect to measure $\mu_i$ pertaining to the dynamics on each manifold. 
    \item Let $K_i$ be the number of distinct Lyapunov exponents of $(T_i,\mu_i)$ with indices indexed $q=1,\ldots,K_i$ with associated lyapunov exponents $\lambda_q^i$. 
    %and $\lambda_q^k$ we implicitly set $K_i=K_k=: K$ (padding the shorter spectrum with $-\infty$ otherwise).}

    \item The cocycle on each $\mathcal{M}_i$ admits an Oseledets splitting $\mathbb{R}^{2m+1} = \bigoplus_q E_q^i$ and noise is isotropic across Oseledets subspaces (expected for uncorrelated thermal/instrument noise). That is to say each Oseledets subspace gets the same order of noise. Projections $\varepsilon_i^{(q)}\in E_q^i$ of $\varepsilon_i$ then satisfy  $\mathbb{E}\|\varepsilon_i^{(q)}\|^2\asymp \sigma/K$ uniformly in $i,q$ where $\sigma=\mathrm{tr}(\,\mathrm{Cov}(\Omega))$, and $\|\varepsilon_i\|^2=\sum_q\|\varepsilon_i^{(q)}\|^2 + o(\sigma)$ in the small-noise regime. 
    \item $\mathcal{M}_i$ is compact. The backmap $\hat{\phi}:\mathcal{M}_i\to\mathcal{M}$ admits a uniform Lipschitz constant $L > 0$ with respect to the Euclidean metrics induced on $\mathcal{M}_i\subseteq\mathbb{R}^{2m+1}$ and $\mathcal{M}\subseteq\mathbb{R}^{2m+1}$.
    \item Candidate observables yield componentwise-ordered Lyapunov spectra: if $h^{KS,UB}_i \leq h^{KS,UB}_k$ then $\lambda_q^i \leq \lambda_q^k$ for every $q\in\{1,\ldots,K\}$.
 \end{enumerate}

$h^{KS,UB}_i:=\sum_{q:\lambda_q^i>0}\lambda_q^i$ is the Ruelle upper bound~\cite{Ruellesinequality} on the effective Kolmogorov--Sinai entropy of $(T_i,\mu_i)$. 

% the corresponding finite-resolution estimator at scale $\|\varepsilon_i\|$ is provided by the Eckmann--Ruelle procedure. Comparisons $h^{KS,UB}_i \leq h^{KS,UB}_k$ are meaningful only when their estimator uncertainties $\Delta h^{KS}$ satisfy $\Delta h^{KS} \ll |h^{KS,UB}_i - h^{KS,UB}_k|$.} 

Let $\epsilon_i$ be the lowest possible mean root mean squared error in reconstruction of all points $r(t)\in\mathcal{M}$ from                 
  $r_i(t)\in\mathcal{M}_i$ performed by any backmapping procedure $\hat{\phi}$ that is locally Lipschitz continuous about any point              
  $r_i(t)\in\mathcal{M}$ with $L>0$.  For prediction horizon $\tau$,                                                                       
  \begin{equation}
   \epsilon_i(\tau) \;=\; \inf_{\hat\phi}\,
     \Big(\mathbb{E}_{r(0)\sim\mu}\,\bigl\|r(\tau)-\hat\phi\!\circ\! T_i(\tau)\,r_i(0)\bigr\|^{2}\Big)^{1/2},
  \end{equation}
  where $r_i(0)\in\mathcal{M}_i$ is the Takens image of $r(0)\in\mathcal{M}$ and the infimum ranges over $\hat\phi\in\mathrm{Lip}_L(\mathcal{M}_i,\mathbb{R}^{2m+1})$, the class of Lipschitz-$L$ backmaps. This class is equicontinuous and uniformly bounded on the compact $\mathcal{M}_i$, hence precompact in $C(\mathcal{M}_i,\mathbb{R}^{2m+1})$ by Arzel\`a--Ascoli, so the infimum is attained. In practice, $\hat\phi$ is realized by a neural-network
  approximator. Smoothness imposes an a posteriori Lipschitz constant which gives an empirical $L$ for a given prediction problem. 

\begin{theorem}
From these assumptions it follows that for any $i,k\in\{1,\ldots,n\}$,
\[
 h^{KS,UB}_i \;\leq\; h^{KS,UB}_k
 \quad\Longrightarrow\quad
 \limsup_{\tau\to\infty}\bigl(\epsilon_i(\tau)-\epsilon_k(\tau)\bigr)\leq 0,
\]
with strict inequality whenever $h^{KS,UB}_i < h^{KS,UB}_k$.
\end{theorem}

\begin{remark}\label{KSE:rmk:absent_iii}
Without componentwise spectrum ordering, $h^{KS,UB}_i$ we rely purely in the sum of positive exponents which need not correspond to equivalent ordering on $\{\epsilon_i(\tau)\}$. 
\end{remark}
\begin{remark}\label{KSE:rmk:regime} The $\limsup$ captures asymptotic separation. Divergence in trajectories is led by the leading Lyapunov exponent until saturation at the attractor diameter $\mathrm{diam}(\mathcal{M})$. The theorem is therefore most informative on horizons where we have not yet reached the attractor saturation time $\tau^{*}\sim(\lambda_{\max}^i)^{-1}\log(\mathrm{diam}(\mathcal{M})/\|\varepsilon_i\|)$. That is to say for $\tau \sim \mathcal{O}(1/\lambda_{\max}^i)$.  If $|h^{KS,UB}_i - h^{KS,UB}_k| \leq \Delta h^{KS}$, the ordering is not predictive.
\end{remark}

\begin{proof}

As the orbits $T_i(t)$ are ergodic with respect to their invariant measures $\mu_i$, Birkhoff's Ergodicity Theorem guarantees that time averages computed along these orbits are equivalent to ensemble averages. This justifies the use of the Kolmogorov--Sinai entropies and the application of Oseledets Multiplicative Ergodic Theorem that follow.

\textbf{Step 1:} Establishing diffeomorphic relationships
\newline

Oseledets Multiplicative Ergodic Theorem (OMET) (Sect.~\ref{KSE:sec:omet}) provides the convergence criterion in the rates of growth or shrinkage of the tangent space (or bundle) at almost every point on the attractor $\mathcal{M} \subseteq \mathbb{R}^{2m+1}$. For some linear operator $A_t(r) = A\bigl(T^{t-1}(r)\bigr) \cdots A\bigl(T(r)\bigr) A(r)$ operating on the orbit $T^j(r)$ acting on vector $r(t)\in\mathcal{M}\subseteq\mathbb{R}^{2m+1}$, OMET tells us that we can subdivide the tangent space of the attractor $\mathcal{M} \subseteq \mathbb{R}^{2m+1}$ into a filtration of subspaces $\{0\}\subseteq V_0\subseteq V_1...V_q\subseteq V_{q+1}...\subseteq V_K\subseteq\mathbb{R}^{2m+1}$for which there exist a corresponding set of Lyapunov exponents $ \lambda_1\geq \lambda_2 \geq ... \geq \lambda_K$ measuring rate of growth or shrinkage of each subspace under the transformation $A$:

\begin{equation}
    \lim_{t\rightarrow\infty}\dfrac{1}{t}\log\|A_t(r)v_q\|= \lambda_q
\end{equation}

where $v \in V_q\backslash V_{q-1}$ is a nonzero vector in subspace $V_q$ but not $V_{q-1}$.

We can equivalently apply OMET to $\mathcal{M}_{i} \subseteq \mathbb{R}^{2m+1}$ constructed from observable $\mathrm{O}_i(t)$ via Takens' Delay Embedding Theorem, leading to its own filtration of subspaces $\{0\}\subseteq V_0^{i}\subseteq V_1^{i}...V_q^{i}\subseteq V_{q+1}^{i}...\subseteq V_K^{i}\subseteq\mathbb{R}^{2m+1}$ for which there exist a corresponding set of Lyapunov exponents $ \lambda^{i}_1\geq \lambda^{i}_2 \geq ... \geq \lambda^{i}_K$.

Takens' theorem guarantees that the manifolds $\mathcal{M}$ and $\mathcal{M}_{i}$ are diffeomorphically related; that is to say there exists a smooth, invertible mapping between the two. Moreover, there exists a bijective map between any given point on the orbits $T(r)$ and $T^{i}(r_{i})$ in the limit that the requirements of Takens' theorem have been fulfilled. The induced diffeomorphism between $\mathcal{M}$ and $\mathcal{M}_{i}$ holds for their tangent bundles; this ensures invariant splittings defined by the filtration in OMET also enjoy a diffeomorphic relation. Consequently the reconstruction of any given invariant subspace $V_q$ will be a function of the deformation to the intrinsic geometry of each subspace induced by noise and forcing. That is to say, the sensitivity of each subspace $V_q^{i}$ to infinitesimal errors in our estimation of $v_q^{i} \in V_q^{i}$ is given by $\lambda_q^{i}$.

\textit{Summary: Provided the requirements of Takens' Delay Embedding Theorem are met, Oseledets multiplicative ergodic theorem guarantees that both the true intrinsic manifold describing the dynamics of $\mathbb{S}$, $\mathcal{M}$ and its reconstruction from time series $\mathrm{O}_i(t)$ are diffeomorphically related as are the subspaces in their OMET filtrations.}

\textbf{Step 2:} Relating Lyapunov exponents to error sensitivity
\newline

Under the noise model of the setup, reconstruction of the embedding $\mathcal{M}_i$ via Takens' Delay Embedding Theorem propagates the perturbations $\varepsilon_i(t)$ into errors in reconstructed states. We denote the true state at time $t$ as $r(t)\in\mathcal{M}$ and the state on the reconstructed attractor $\mathcal{M}_i$ as $r_i(t)$. The pointwise reconstruction error is then:

$$ \epsilon_i(t) = \|r(t) - \hat{r_i}(t)\|$$

where $\hat{r_i}(t)$ is the output of the backmapping from $r_i(t)$ to $r(t)$.

We consider a linearization of these dynamics around a given point $r_i(t)$. Let $A(t)$ be a linear operator propagating $r_i(t)$ to a new $r_i(t+\delta t)$. Then a small perturbation to $r_i(t)$, $\delta(t)$ will result in perturbation $$\delta(t+\tau)\approx A(t)\delta(t)$$.

Given an initial perturbation $\|\delta r(0)\| = \|\varepsilon_i(0)\|$ and assuming a single positive Lyapunov exponent $\lambda$,
$$\|\delta r(\tau)\| \approx \|\varepsilon_i(0)\|\, e^{\lambda \tau}$$%
%\cite{FarmerSidorowich1987}.

Extending this to $\mathcal{M}_i$ with its $K$ Lyapunov exponents, $\lambda_q^i$, corresponding to the $q^{th}$ invariant subspace $V_q^i$, then the sensitivity of dynamics in each subspace is given by:

$$\|\delta^i_q(\tau)\| \approx \|\varepsilon_q(0)\| \, e^{\lambda_q^i \tau}$$

where $\|\varepsilon_q(0)\|$ describes the magnitude of the initial perturbation in subspace $V_q^i$. Summing over all subspaces, the total error is given by the Euclidean distance between the perturbed and unperturbed states:

$$\|\delta_{\text{total}}^i(\tau)\| \approx \sqrt{\sum_{q=1}^K \left(\|\varepsilon_q(0)\| \, e^{\lambda_q^i \tau}\right)^2}$$

This effect is dominated by those subspaces whose $\lambda_q^i>0$ as these experience exponential divergence. 

We can then relate the reconstruction error $ \epsilon_i$ to $\delta_{\text{total}}^i(\tau)$. As the sum of the positive Lyapunov exponents (KSE) decreases, so too does $\delta_{\text{total}}^i(\tau)$. In the limit of Takens' Theorem holding, $\mathcal{M}_i$ and $\mathcal{M}$ enjoy a diffeomorphic relationship. The limitation to learning this mapping with a universal functional approximator in this limit comes from errors in the estimation of any $r_i(t)\in\mathcal{M}_i$. Then the reconstruction error $\epsilon_i$ is a function of $(\delta_{\text{total}}^i(\tau))$.

\textit{Summary: For the noisy time series $\tilde{\mathrm{O}}_i(t)$, points on the reconstructed manifold $\mathcal{M}_i$ will be reconstructed with mean error $\|\delta_{\text{total}}^i(\tau)\| \approx \sqrt{\sum_{q=1}^K \left(\|\varepsilon_q(0)\| \, e^{\lambda_q^i \tau}\right)^2}$} 

\textbf{Step 3:} Error in the backmapping of $r_i(t)$
\newline

There exists a smooth, continuous, invertible and diffeomorphic mapping between the manifolds $\mathcal{M}_i \to \mathcal{M}$ given by Takens' Theorem:

$$\phi: \mathcal{M}_i \to \mathcal{M} $$

approximated by continuous backmapping function $\hat{\phi}: \mathcal{M}_i \to \mathcal{M}$. Then for every $\epsilon > 0 $ $ \exists$ $\delta> 0$ such that for any $r_i(t) \in \mathcal{M}_i$ perturbation $\delta_{\text{total}}^i $ with

$$\|r_i(t) - \bigl(r_i(t) + \delta_{\text{total}}^i\bigr)\| < \delta$$

when $r_i(t)$ and $r_i(t)+\delta_{\text{total}}^i \in \mathcal{M}_i$ are mapped onto $\mathcal{M}$, they will be at some distance

$$\|\hat{\phi}(r_i(t)) - \hat{\phi}\bigl(r_i(t) + \delta_{\text{total}}^i\bigr)\| < \epsilon$$

away from each other. Then since $\hat{\phi}$ is uniformly Lipschitz on the compact manifold $\mathcal{M}_i$ with Lipschitz constant $L > 0$ (an assumption of the proof and generally true for neural network architectures):

$$\|\hat{\phi}(r_i(t)) - \hat{\phi}\bigl(r_i(t) + \delta_{\text{total}}^i\bigr)\| \le L \|\delta_{\text{total}}^i\|$$

Consequently, small errors in the estimation of $r_i(t)$ translate monotonically into small errors in the reconstructed state $\hat{\phi}(r_i(t))$.

\textit{Summary: Small errors in estimation of  $r_i(t)$, $(\delta_{\text{total}}^i(\tau))$ result in reconstruction errors in the estimation of $\hat{\phi}(r_i(t))$ that are monotonically related to $(\delta_{\text{total}}^i(\tau))$}

\textbf{Step 4:} Bounding the mean error, $\epsilon_i$

We define the mean reconstruction error \(\epsilon_i\) as the mean of the RMSEs computed for each frame. That is,

$$\epsilon_i = \frac{1}{N} \sum_{t=1}^N \|r(t) - \hat{r}_i(t)\| = \frac{1}{N} \sum_{t=1}^N \epsilon_i(t)$$

where $N$ is the total number of time frames over which the reconstruction is computed.

Then substituting the error bound for a single time step in \textbf{Step 3}, $\epsilon_i $ evaluates to

$$\|r(t) - \hat{r}_i(t)\| = \|\hat{\phi}(r_i(t)) - \hat{\phi}(r_i(t) + \delta_{\text{total}}^i(t))\| \le L\, \|\delta_{\text{total}}^i(t)\| $$

The overall reconstruction error \(\epsilon_i\) can be bounded by the average of these local errors:

$$\epsilon_i \le \frac{L}{N} \sum_{t=1}^N \|\delta_{\text{total}}^i(t)\| $$

In \textbf{Step 2}, we showed that $$\|\delta_{\text{total}}^i(\tau)\| \approx \sqrt{\sum_{q=1}^K \left(\|\varepsilon_q(0)\| \, e^{\lambda_q^i \tau}\right)^2}$$ 

By assumption, $h^{KS,UB}_i \leq h^{KS,UB}_k$ implies $\lambda_q^i \leq \lambda_q^k$ componentwise and $|\varepsilon_i^{(q)}|$ are comparable across observables, justifying a term-by-term inequality. Consequently, $\|\delta_{\text{total}}^i(t)\|$ is componentwise-monotone in the positive $\lambda_q^i$. This implies then that  $h^{KS,UB}_i = \sum_{q: \lambda_q^i > 0} \lambda_q^i$ inherits that monotonicity. Then, if

$$h^{KS,UB}_i \leq h^{KS,UB}_k$$
then

$$\epsilon_i\leq \epsilon_k$$

\textbf{Conclusion:} Then if $h^{KS,UB}_i \leq h^{KS,UB}_k$, $\epsilon_i\leq \epsilon_k$.

\end{proof}

\subsection{Limitations}

We recognize the following limitations to this theorem and proof. We assume that errors along different invariant subspaces combine in a Euclidean manner. This assumption of local linearizability will break down under strong enough nonlinearity or large enough noise \cite{strogatz,KantzSchreiber2004}. We have treated noise as being separable from dynamics. In many cases this assumption holds, however if noise interacts with the system's nonlinear dynamics, the linearized error propagation may not be sufficient to bound errors in recorded $r_i(t)$. Examples of this include accrued errors from the propagation of differential equations over long time periods with numerical means \cite{frenkel2001understanding} and experimental observations with error that is a function of the recorded observable such as single-molecule F\"orster resonance energy transfer (smFRET)~\cite{Roy2008rr}.

We reiterate that the proof applies to functional approximators respecting this condition of local Lipchitz continuity. We note that this condition is required for any smooth and diffeomorphic map between $\mathcal{M_i}$ and $\mathcal{M}$ \cite{rudin1976pma,spivak1965com}. While we cannot make such a guarantee for any arbitrary universal functional approximator, if such an approximator is close to the true map, it must also respect this condition of a smooth and diffeomorphic mapping as guaranteed by Takens' Theorem.

Furthermore, observables with comparable $h^{KS,UB}$ but differently-concentrated spectra (i.e.\ different distributions of positive Lyapunov exponents) can produce divergent reconstruction errors. The assumption of componentwise ordering --- $h^{KS,UB}_i \leq h^{KS,UB}_k$ implies $\lambda_q^i \leq \lambda_q^k$ for all $q$ --- causes this convexity slack to vanish, and partial violations of this assumption bound our empirical correlations below $\rho = 1$. 

Finally, as noted earlier, under appropriate conditions (e.g., for smooth hyperbolic systems with an SRB measure), Pesin's Theorem shows that the Ruelle Inequality actually becomes an equality \cite{Ruellesinequality,PesinTheorem,Young2002}.

\section{Empirical Validation}\label{KSE:sec:empirical}

The mathematical proof provided above is illustrative of the power that observable choice has in modeling the dynamics of arbitrary dynamical systems under a series of modest assumptions. We have shown that the $h^{KS,UB}_i$ associated with arbitrary scalar observations of chaotic systems can be used to select optimal observables to model the dynamics of the driving process for any smooth reconstruction method. In this section we demonstrate this method and its limits in sample systems of increasing complexity parametrized both by intrinsic dimensionality and noise (Fig.~\ref{fig:attractor_deformation}).

\begin{figure*}[t]
  \centering
  \includegraphics[width=\textwidth]{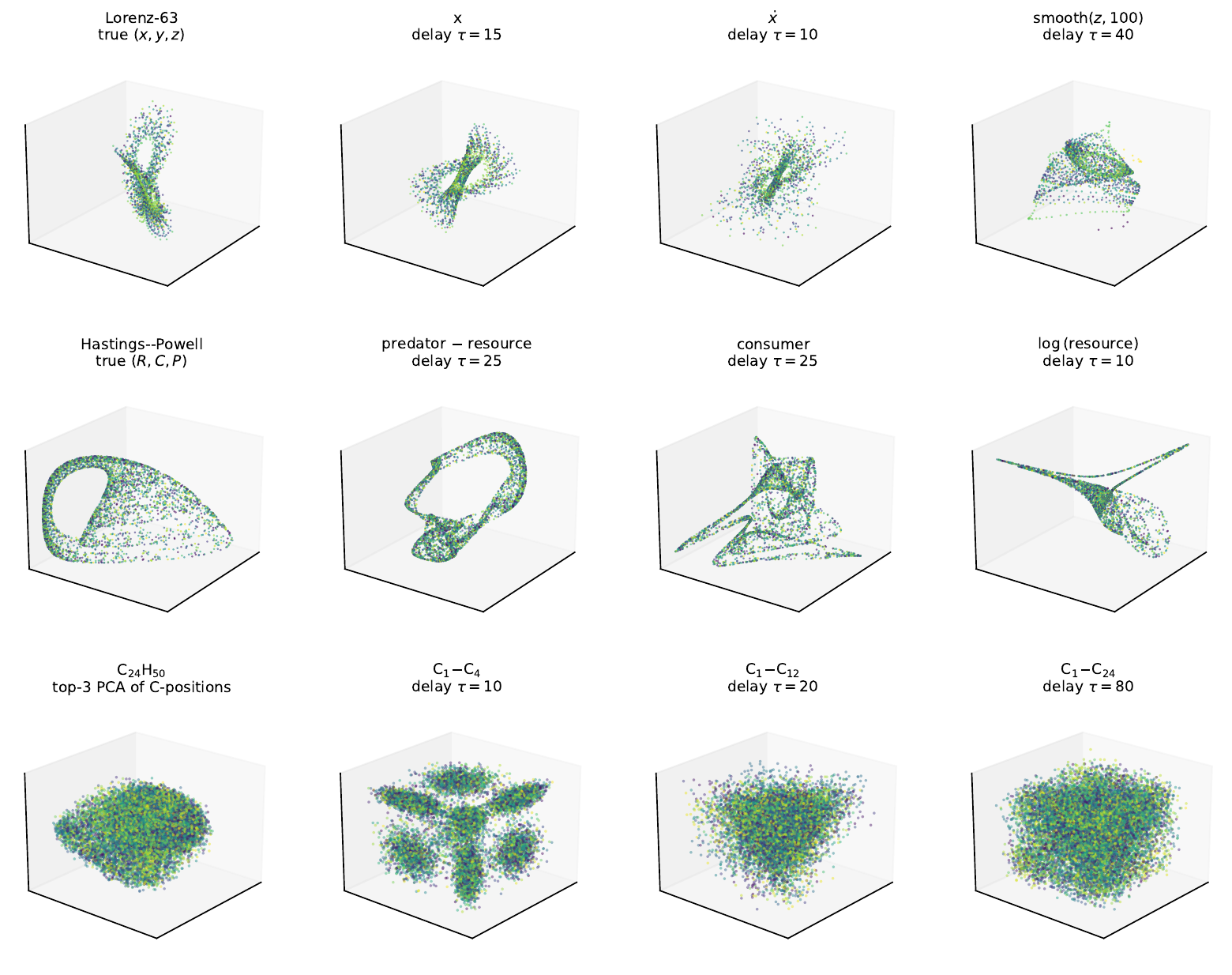}
  \caption{Representations of full dynamics and Takens-reconstructed attractors of the Lorenz-63, Hastings--Powell and Tetracosane systems with various observables. While all share underlying dynamics from Takens' Theorem, different observable choices visibly alter attractor geometry bounded by $h^{KS,UB}$ according to this theorem.}
  \label{fig:attractor_deformation}
\end{figure*}

\subsection{Systems and Parameters}

As this is an empirical test of the theorem, consequently we are using estimations of KSE and the Lyapunov spectrum estimated by the Eckmann-Ruelle algorithm \cite{LyapEckmann} using the nolds package \cite{nolds}. We perform the reconstructions using the TAR pipeline \cite{topel2025}, a neural-network--based reconstruction algorithm for arbitrary time series. We use the 3 examples here to present a variety of regimes. For systems where our sampling is so complete that there exists a diffeomorphism between our observables, the true $h^{KS,UB}$ values are attainable from any arbitrary observable and hence the theorem predicts no ordering. We show this with Lorenz-63 ($\sigma{=}10$, $\rho{=}28$, $\beta{=}8/3$), a canonical low-dimensional chaotic attractor. Secondly, at a sufficiently high level of noise, all observables saturate to the noise signal, destroying any signal from empirical measurement of the KSE. We demonstrate this with the Hastings--Powell three-species food chain \cite{hastings1991chaos}, a 3D model of ecological chaos subject to a variety of Gaussian noise levels. Finally, we introduce a molecular dynamics trajectory of Tetracosane ($\mathrm{C}_{24}\mathrm{H}_{50}$), a complex molecular system at experimental scale, to demonstrate the theorem's applicability to real world problems. 

In order to test this theory, we must introduce a number of observations. For the Lorenz-63 system, we have 20 observables: $x$, $y$, $z$, smoothed-$x$ (windows 5, 25, 100), smoothed-$z$ (windows 5, 25, 100), $\dot x$, $\dot z$, $\ddot x$, $x^2$, $y^2$, $z^2$, $\log(1{+}z)$, $x{+}y$, radial $\sqrt{x^2{+}y^2{+}z^2}$, $xy$, $xz$. For Hastings--Powell we have 20 observables: resource, consumer, predator, $\log$ of each species, smoothed-resource (windows 5, 25, 100), smoothed consumer and predator (window 25), $\dot R$, $\dot C$, $\dot P$, total biomass, $\log$ of total biomass, resource$^2$, predator $-$ resource, resource$\cdot$consumer, predator$/$consumer. Each system was swept densely across an SNR range tailored to its noise sensitivity. For Lorenz-63 and Hastings--Powell we sweep 46 SNR levels: a coarse high-SNR tail $\{\infty, 1000, 700, 500, 400, 300, 250, 200, 150\}$\,dB and a denser grid covering $100$\,dB to $-10$\,dB at $\sim$3\,dB spacing, traversing the noiseless, discrimination, and saturation regimes. For $\mathrm{C}_{24}\mathrm{H}_{50}$, our 21 observables are intramolecular distances $\mathrm{C}_1$--$\mathrm{C}_n$ for $n{=}4,\ldots,24$; the MD trajectory itself already carries thermal noise so we omit the $\mathrm{SNR}=\infty$ baseline and instead sweep 39 levels from $90$\,dB down to $-30$\,dB, extending lower than the ODE systems to resolve the saturation transition (which occurs at substantially lower SNR for this high-dimensional system). Negative-SNR rows are retained in the deposited CSVs but correspond to a regime in which noise exceeds signal and falls outside the discrimination regime that is the focus of our analysis. The shortest distances $\mathrm{C}_1$--$\mathrm{C}_2$ (the C--C bond length) and $\mathrm{C}_1$--$\mathrm{C}_3$ (the 1--3 distance fixed by bond angle) are excluded from the analysis: they probe sub-picosecond bond-stretching and angle-bending modes rather than the slow nanosecond chain reorganizations that the rest of the chain-distance observables capture, and the Eckmann--Ruelle estimator cannot resolve their KSE within our sampling rate. Details of the molecular dynamics simulation are given in Ref.~\cite{ferguson2010systematic}. We perform 3 trials of each to bootstrap error estimates.

\paragraph{Trajectory lengths and integrators.} Lorenz-63 and Hastings--Powell were integrated with the explicit Runge--Kutta DOP853 scheme \cite{HairerNorsettWanner1993} in \texttt{scipy.integrate.solve\_ivp} (rtol $10^{-10}$, atol $10^{-12}$) from initial condition $(1,1,1)$ for Lorenz over $t\in[0, 200]$ with $\Delta t = 0.01$ (20{,}000 samples; first 5{,}000 discarded as transient), and from $(0.7, 0.2, 8.0)$ for Hastings--Powell over $t\in[0, 20{,}000]$ with $\Delta t = 1$ (20{,}000 samples; first 5{,}000 discarded). The $\mathrm{C}_{24}\mathrm{H}_{50}$ trajectory of Ref.~\cite{ferguson2010systematic} is a 500\,ns molecular-dynamics simulation in explicit solvent; we read it at $\mathrm{stride}=25$ to give 20{,}000 samples for direct comparison with the ODE systems.

Code is available on GitHub at \texttt{maxtopel/TAR} while scripts and parameters are available at \texttt{maxtopel/KSE\_codebase}. All reconstruction runs were executed on Northwestern University's Quest HPC cluster. Each job in this study used 4 CPU cores and 16--32 GB RAM.

\subsection{Reconstruction with TAR}

We use the Takens Reconstruction (TAR) pipeline to take time series observations of arbitrary dynamical systems and reconstruct their full dimensional dynamics given some training data \cite{topel2025}. First we standardize observables (fit on first half of the data), then compute the time-delay $\tau$ as the first crossing of the autocorrelation function below a fixed threshold of $0.2$ (a standard practitioner approach to delay-embedding parameter selection \cite{KantzSchreiber2004}) and embedding dimensionality via the Cao $E1(d)$ criterion \cite{Cao1998} to get Takens embeddings. Latent space representations of these embeddings are produced via diffusion maps over some subsample of the training data and the rest are projected on to the latent space manifold extracted via diffusion map. Backmapping to the full dimensional state is performed via \texttt{tir.train\_and\_evaluate\_model2} using a neural network of architecture $5{\times}50$ ReLU trained over 250 epochs with Adam optimizer, learning rate $\eta=10^{-2}$ and 3-fold cross-validation. The KSE upper bound $h^{KS,UB}_i$ was estimated via \texttt{nolds.lyap\_e} (Eckmann--Ruelle QR algorithm \cite{LyapEckmann}) on the scalar time series.

\subsection{Results}

\begin{figure*}[t]
  \centering
  \includegraphics[width=\textwidth]{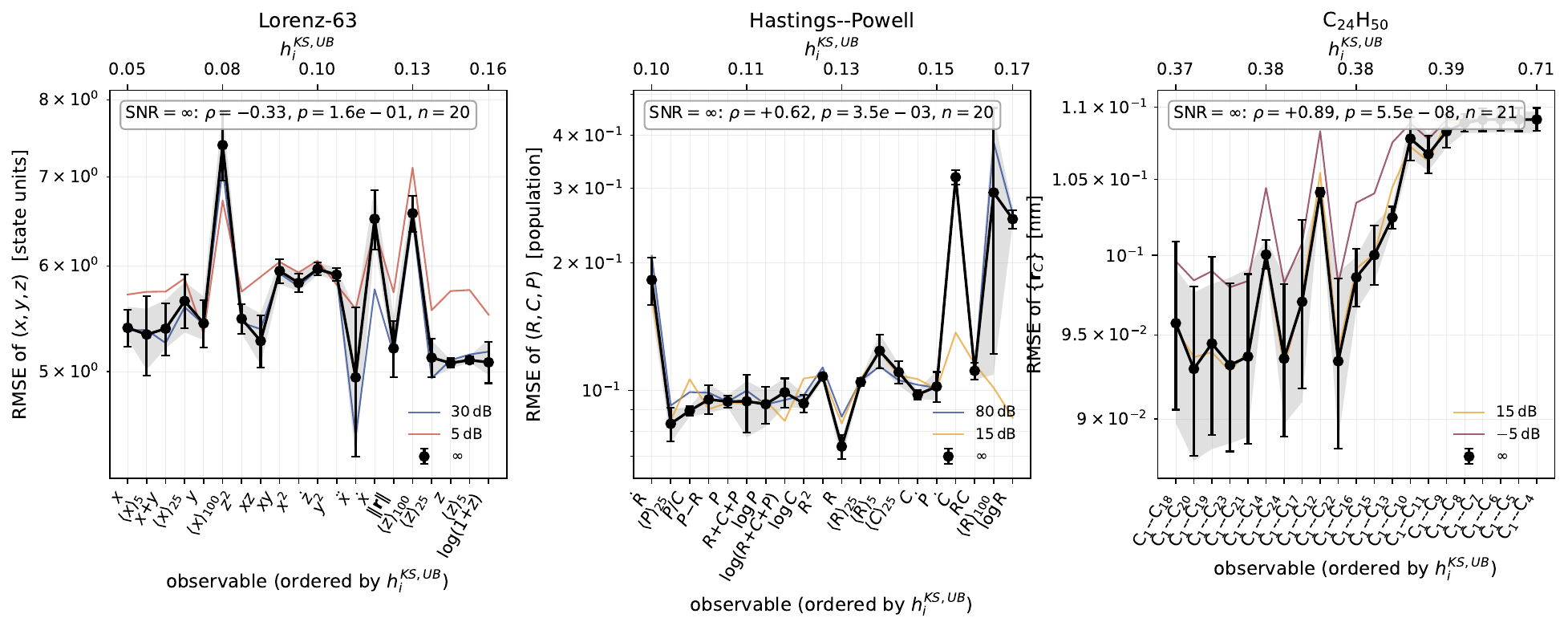}
  \caption{Reconstruction RMSE as a function of scalar observables ordered by estimated Kolmogorov--Sinai upper bound $h^{KS,UB}_i$. Per panel three SNR levels are shown, chosen to span each system's regimes: noiseless ($\infty$, black, with shaded 95\,\% bootstrap CI), the SNR where Spearman $\rho$ peaks for that system, and an SNR in the system's saturation regime. \emph{Lorenz-63} (left, 20 observables): $\infty$ (observable equivalence, $\rho{=}{-}0.33$, $p{=}0.16$), $30$\,dB (discrimination peak, $\rho{=}{+}0.62$, $p{=}0.0039$), $5$\,dB (saturation, $|\rho|{<}0.2$). \emph{Hastings--Powell} (middle, 20 observables): $\infty$ (clean discrimination, $\rho{=}{+}0.62$, $p{=}0.0035$), $80$\,dB (discrimination peak), $15$\,dB (mid-noise breakdown, $\rho{=}{-}0.45$). \emph{$\mathrm{C}_{24}\mathrm{H}_{50}$} (right, 21 chain-distance observables): $\infty$ ($\rho{=}{+}0.89$, $p{=}5.5{\times}10^{-8}$), $15$\,dB (sharpened peak, $\rho{=}{+}0.97$, $p{=}10^{-13}$), $-5$\,dB (saturation onset --- the noise-amplified RMSE at high-KSE observables visibly elevates the curve). RMSE axes are log-scale.}
  \label{fig:rmse_vs_kse}
\end{figure*}

Per-observable RMSE for representative SNR levels is shown in Fig.~\ref{fig:rmse_vs_kse}, with each panel ordered by estimated $h^{KS,UB}_i$. The three systems each illustrate a different combination of the three regimes predicted by the theorem (observable equivalence, discrimination, and saturation; cf.\ Remark~\ref{KSE:rmk:regime} and Fig.~\ref{fig:bell_curves}).

\textbf{Lorenz-63 --- traversing all three regimes.} At clean SNR, near-complete sampling of the low-dimensional attractor produces a diffeomorphism between the manifolds reconstructed from each observable, placing the system in the \emph{observable-equivalence regime}: all coordinate choices yield uniformly good Takens reconstructions, and we find no significant relationship between KSE and reconstruction error ($\rho{=}{-}0.33$, $p{=}0.16$, $n{=}20$). Adding moderate noise breaks this equivalence: the reconstruction quality of high-KSE observables degrades faster than that of low-KSE ones, and the system enters the \emph{discrimination regime}, recovering the rank ordering predicted by the theorem ($\rho_{30\,\mathrm{dB}}{=}{+}0.62$, $p{=}0.0039$; $\rho_{20\,\mathrm{dB}}{=}{+}0.46$, $p{=}0.043$). At higher noise ($\mathrm{SNR}\le 10$\,dB), the correlation flattens into the \emph{saturation regime}, with $|\rho|<0.2$ across SNRs $\in\{10, 5, 0\}$\,dB. There the Eckmann--Ruelle Lyapunov-exponent estimates no longer reflect the true KSEs of the underlying dynamics; the $h^{KS,UB}$ axis collapses, exactly the regime condition $|h^{KS,UB}_i - h^{KS,UB}_k| \lesssim \Delta h^{KS}$ flagged in Remark~\ref{KSE:rmk:regime}.

\textbf{Hastings--Powell --- discrimination at clean, then saturation.} The 3D ecological dynamics are diverse enough that no observable-equivalence regime exists at clean SNR: candidate observables already differ meaningfully in $h^{KS,UB}$ and the system sits in the discrimination regime, with a statistically significant positive correlation matching the theorem ($\rho{=}{+}0.62$, $p{=}0.0035$, $n{=}20$). Any added noise drives the system monotonically toward saturation: the rank ordering collapses by 30\,dB ($\rho{=}{-}0.05$, $p{=}0.84$) as additive noise smooths the RMSE landscape across observables.

\textbf{$\mathrm{C}_{24}\mathrm{H}_{50}$ --- discrimination across the entire physical SNR range.} The molecular system is the practical case of interest: a high-dimensional attractor reconstructed from observables (intramolecular distances) accessible via smFRET experiments. In principle all observations describe the same underlying molecular dynamics, but the high dimensionality and absence of canonical-coordinate equivalence place every observable in the discrimination regime even at clean SNR. We find a Spearman rank correlation $\rho(h^{KS,UB},\mathrm{RMSE}) = {+}0.892$ ($p{=}5.5{\times}10^{-8}$, CI: $[{+}0.67, {+}0.97]$) over all 21 observables, with the shortest distance $\mathrm{C}_1$--$\mathrm{C}_4$ having the largest KSE (0.71) and largest RMSE (0.109\,nm) and the longest $\mathrm{C}_1$--$\mathrm{C}_{24}$ the smallest KSE (0.38) and smallest RMSE (0.094\,nm). Strikingly, the correlation \emph{sharpens monotonically} with added measurement noise --- $\rho_{30\,\mathrm{dB}}{=}{+}0.92$, $\rho_{20\,\mathrm{dB}}{=}{+}0.94$, $\rho_{15\,\mathrm{dB}}{=}{+}0.97$, peaking at $\rho_{10\,\mathrm{dB}}{=}{+}0.97$ (all $p<10^{-9}$) --- because additional noise selectively penalizes high-KSE observables, amplifying rather than degrading the theorem's signal. The correlation remains in this near-perfect plateau ($\rho \gtrsim 0.95$) down to $\mathrm{SNR}\approx -3$\,dB, and only enters saturation when noise exceeds signal by $\gtrsim 10\times$ ($\mathrm{SNR}\lesssim -3$\,dB), where the KSE estimator is finally overwhelmed.

\begin{figure*}[t]
  \centering
  \includegraphics[width=0.85\textwidth]{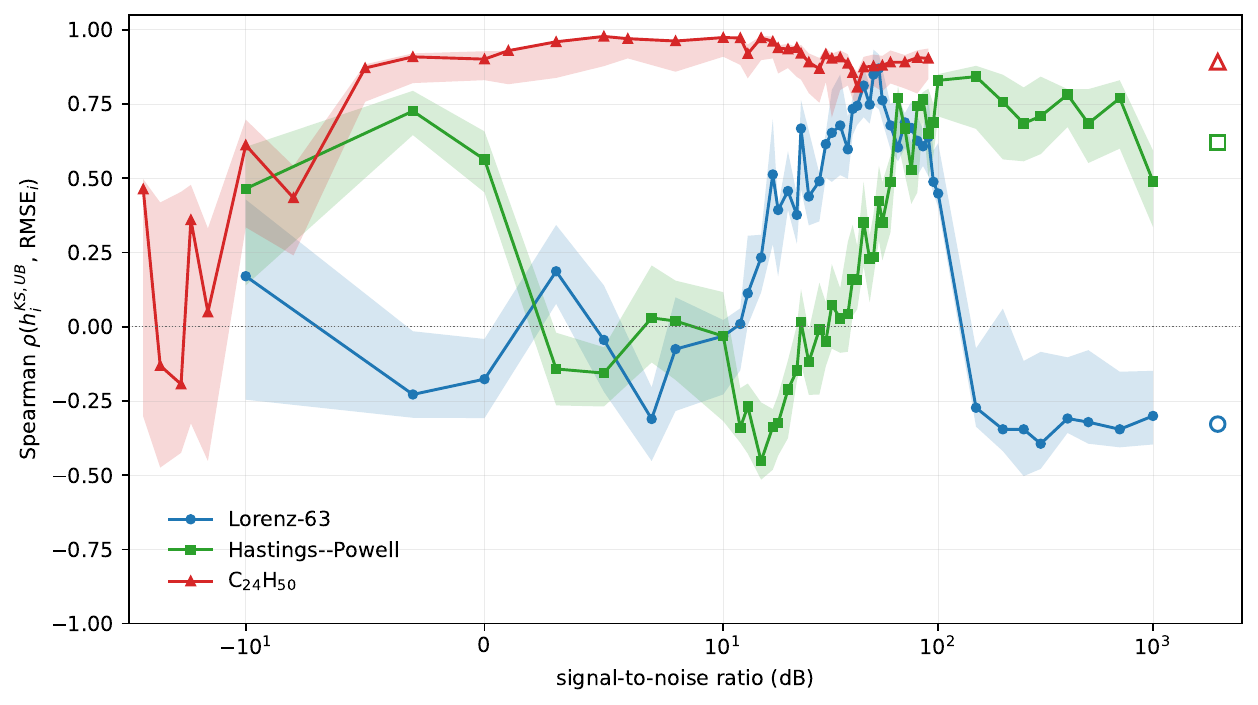}
  \caption{Spearman rank correlation $\rho(h^{KS,UB}_i, \mathrm{RMSE}_i)$ as a function of measurement-noise SNR for each system, with bootstrap 95\% confidence intervals. The theorem predicts three regimes, demonstrated by the three systems. At high SNR ($\gtrsim 100$\,dB) Lorenz-63 and Hastings--Powell demonstrate a strong degree of observable equivalence, a regime in which all observables construct faithful Takens manifolds and consequently perform reconstruction well. As noise increases, we enter a discrimination regime where the theory is most informative. It is here that the effect of noise on attractor representations is addressed by this work. At low SNR the KSE estimator saturates and the correlation collapses to zero. $\mathrm{C}_{24}\mathrm{H}_{50}$, lacking the canonical-coordinate equivalence regime entirely, maintains $\rho \gtrsim 0.9$ across an enormous SNR range and only saturates once noise exceeds signal by $\gtrsim 10\times$ ($\mathrm{SNR}\lesssim -3$\,dB). Open markers at the right edge denote the noiseless ($\mathrm{SNR}=\infty$) limit.}
  \label{fig:bell_curves}
\end{figure*}

\textbf{Synthesis.} Across all three systems (Fig.~\ref{fig:bell_curves}), the empirical evidence recapitulates the theorem and its regime condition (Remark~\ref{KSE:rmk:regime}): the KSE is a reliable observable-selection criterion in the discrimination regime, where candidate observables differ meaningfully in $h^{KS,UB}$ and the estimator resolves those differences. The criterion becomes inactive when observables are too similar (the equivalence regime, only seen for low-dimensional well-sampled systems) and breaks down when noise overwhelms the estimator (the saturation regime, reached at different SNRs for different systems). The high-dimensional MD case ($\mathrm{C}_{24}\mathrm{H}_{50}$) is the practitioner-relevant scenario: the discrimination regime occupies the entire physically meaningful SNR range and the criterion is at its strongest under realistic measurement noise.

\section{Conclusion} \label{KSE:sec:conc}

We have now proven that manifolds constructed from time series with greater Kolmogorov-Sinai entropies admit a higher upper bound on mean reconstruction error under modest technical conditions. Section~\ref{KSE:sec:empirical} demonstrates the practical implications of our proof in noiseless and noisy conditions. The 21 observables of $\mathrm{C}_{24}\mathrm{H}_{50}$ demonstrate a Spearman rank correlation between $h^{KS,UB}$ and reconstruction RMSE of $+0.89$ ($p{=}5.5{\times}10^{-8}$) at clean SNR, sharpening to $+0.97$ at $10$\,dB SNR making the theorem more predictive under realistic measurement noise in this high-dimensional regime. The toy Hastings--Powell and Lorenz systems demonstrated the observable-equivalence regime, where coordinates describe equivalent representations, and the differential effect of noise as a function of the equivalence of those representations (Lorenz moderate-noise Spearman $\rho_{30\,\mathrm{dB}}{=}{+}0.62$ at $p{=}0.0039$; Hastings--Powell clean-noise $\rho{=}{+}0.62$ at $p{=}0.0035$). The three systems span complementary regimes of the theorem: an observable-equivalence regime where it is inactive (Lorenz, clean), a discrimination regime where it is most informative (Hastings--Powell, clean; C\textsubscript{24}H\textsubscript{50}, all SNRs), and a saturation regime where the estimator can no longer resolve KSEs (all systems at deep noise). This proof represents an analytical description for a physical intuition surrounding observable selection and provides a method to do observable choice optimization \emph{a priori} to guide modeling choice. Its generality allows it to be useful for arbitrary dynamical systems before performing costly modeling.

Bounding problems are abound at the intersection of embedding theorems and statistical learning. We recognize that the requirement of ergodicity for the use of OMET and Takens is a limitation to generality. However, as was demonstrated by our empirical results, it is precisely the limit in which all observations do NOT form equivalent embeddings of dynamics that this proof is most powerful. In future work, we seek to extend this proof to describe the size of these reconstruction errors as a function of sampling time. Moreover, we would like to form a deeper understanding of the nature of topographical deformations in manifolds constructed from Takens' Theorem. Advances in these and similar areas will provide a mathematical foundation for observable selection in complex systems where such problems are beyond physical intuition.

\section*{Acknowledgments}%
The author thanks Andrew L.\ Ferguson (University of Chicago) for the $\mathrm{C}_{24}\mathrm{H}_{50}$ molecular-dynamics trajectory and for foundational discussions on Takens reconstruction over the years. The author also thanks Madhav Mani (Northwestern University) for the continuing education on dynamics that have encouraged this work. This research was supported in part through the computational resources and staff contributions provided for the Quest high-performance computing facility at Northwestern University which is jointly supported by the Office of the Provost, the Office for Research, and Northwestern University Information Technology.

\section*{Data Availability}%
The TAR pipeline used for the empirical validation in Sect.~\ref{KSE:sec:empirical} is openly available at GitHub \texttt{maxtopel/TAR}~\cite{topel2025}. Scripts, parameter files, and result CSVs for Lorenz-63, Hastings--Powell, and $\mathrm{C}_{24}\mathrm{H}_{50}$ are archived at \texttt{https://github.com/maxtopel/KSE\_codebase} and on Zenodo at \href{https://doi.org/10.5281/zenodo.19899204}{10.5281/zenodo.19899204}. The $\mathrm{C}_{24}\mathrm{H}_{50}$ molecular dynamics trajectory was previously published in Ref.~\cite{ferguson2010systematic}.

\bibliographystyle{apsrev4-2}
\bibliography{ref}% Produces the bibliography via BibTeX.

\end{document}